\documentclass[11pt]{article}
\usepackage{fullpage}
\usepackage{amsmath, amsthm, slashed}
\usepackage[left=1in,top=1in,right=1in,bottom=1in,headheight=3ex,headsep=3ex]{geometry}
\usepackage{graphicx}
\usepackage{float}

\usepackage{indentfirst}

\newtheorem{theorem}{Theorem}

\usepackage{amsmath,amssymb,amsfonts,amsthm,amsbsy,amstext}
\usepackage[authordate,backend=b
iber]{biblatex-chicago}
\usepackage{graphicx}
\usepackage{wrapfig}
\usepackage{multicol}
\usepackage{geometry}
\usepackage{amssymb}
\usepackage{mathtools}
\usepackage{siunitx}
\usepackage{color,soul}
\DeclareMathOperator*{\argmax}{arg\,max}

\usepackage{times}
\usepackage[T1]{fontenc}
\usepackage[nodayofweek]{datetime}

\usepackage{setspace}
\usepackage{multicol}
\usepackage{url}
\usepackage{lastpage}
\usepackage{amsmath}
\usepackage{layout}
\usepackage{hyperref}
\usepackage{cleveref}
\usepackage{epigraph} 
\usepackage{algorithm}
\usepackage{algpseudocode}
\usepackage{array, xcolor}
\usepackage{color}
\usepackage{subfig}
\definecolor{clemsonorange}{HTML}{EA6A20}
\hypersetup{colorlinks,breaklinks,linkcolor=blue,urlcolor=blue,anchorcolor=clemsonorange,citecolor=blue}

\newcommand{\pare}[1]{\left(#1\right)}
\newcommand{\brac}[1]{\left[#1\right]}

\newcommand{\dis}{\displaystyle}
\usepackage{url}

\newcommand{\R}{\mathbb R}

\newcommand{\dd}{\partial}

\newcommand{\comment}[1]{}

\newcommand{\Var}{\text{Var}}

\newtheorem{corollary}{Corollary}[theorem]
\newtheorem{lemma}[theorem]{Lemma}

\theoremstyle{definition}

\newtheorem{definition}{Definition}

\crefname{assumption}{assumption}{assumptions}
\Crefname{assumption}{Assumption}{Assumptions}
\crefname{assumptionalt}{assumption}{assumptions}
\Crefname{assumptionalt}{Assumption}{Assumptions}

\usepackage{tkz-graph} 
\usetikzlibrary{arrows,automata, positioning}
\usetikzlibrary{shapes.geometric}%

\usepackage{titlesec}
\titleformat*{\section}{\large\bfseries}
\titleformat*{\subsection}{\bfseries}
\titleformat*{\subsubsection}{\bfseries}
\titleformat*{\paragraph}{\bfseries}
\titleformat*{\subparagraph}{\bfseries}

\title{A Comment on Rational Inattention with Arbitrary Choice Sets}
\author{Christopher W. Engh}

\begin{document}
\maketitle 

\begin{abstract}
    This comment points out that rational inattention is a nested regularized optimal transport problem. We use entropic optimal transport to establish the main results in \textcite{matvejka2015rational,caplin2019rational} and extend them to arbitrary choice sets.
\end{abstract}

\section{Introduction}

At the dawn of the millennium, the pioneering work of \textcite{sims1998stickiness,sims2003implications,sims2005rational} argued that the inertial behaviour of economic variables could be a consequence of optimizing agents with limited information-processing capacity. Initially, rational inattention centered on the case of continuous action spaces (such as $\R$). About a decade ago, \textcite{matvejka2015rational,caplin2019rational} developed a tractable framework which limited the agent to a finite set of actions and established the key results on discrete rational inattention under Shannon information costs. Around the same time, the groundbreaking work of \textcite{cuturi2013sinkhorn} on fast computation of optimal transport via entropy regularization sparked a renewed interest in the \textit{Schr\"odinger Bridge problem}, which has since found far-ranging applications from image processing to generative AI. The purpose of this comment is to show how these problems are connected and how results from entropic optimal transport can be applied to extend the tractable discrete choice rational inattention framework to arbitrary choice sets. \\

Let us start with a quick review of rational inattention. Given a finite set of actions $A$, a set of states $\Omega$, a prior $\mu$ over $\Omega$, and a utility function $u:A\times\Omega\to \R$, the rationally inattentive agent commits to an \textit{information policy} $P(\alpha|\omega)$ mapping states $\omega \in \Omega$ to a distribution over action recommendation signals $\alpha \in A$, with the goal of maximizing expected utility less the Shannon information between $\alpha $ and $\omega$. \\

Mathematically, the problem is to choose a joint distribution $P\in \Delta(A\times\Omega)$, satisfying \textit{Bayes plausibility} $P_\omega=\mu$, which maximizes the objective
\begin{align}\label{MMobjective}
    \max_{P_\omega=\mu} \int \sum_{\alpha\in A}\frac{u(\alpha,\omega)}{\lambda}\ P(\alpha|\omega)\ d\mu(\omega)-D_{KL}(P\|P_\alpha\otimes P_\omega)
\end{align}
where $P_\alpha, P_\omega$ are the respective marginals, and $D_{KL}(\cdot\|\cdot)$ is the \textit{Kullback-Leibler divergence}
\begin{align*}
    D_{KL}(Q\|R)=\begin{cases}
        \int \pare{\frac{dQ}{dR}}\log\pare{\frac{dQ}{dR}}\ dR& \text{if }Q\ll R\\
        +\infty& \text{else}
    \end{cases}
\end{align*}
$D_{KL}(P\|P_\alpha\otimes P_{\omega})\equiv D_{KL}(P\|P_\alpha\otimes \mu)$ is the \textit{Shannon information between $a$ and $\omega$ under $P$}. In words, the payoff of $P$ is the expected utility under $P$ minus the mutual information encoded in $P$.\\

A natural extension to infinite spaces is as follows. Let $A$ and $\Omega$ be Polish spaces, the latter equipped with a prior $\mu\in\Delta\Omega$. Let $u:A\times \Omega\to (-\infty,\overline u]$ be measurable. The agent commits to an information policy $P\in \Delta(A\times \Omega)$ which maximizes the objective
\begin{align}\label{contobjective}
    \sup_{P_\omega=\mu}\int \frac{u(\alpha,\omega)}{\lambda}\ dP- D_{KL}(P\|P_\alpha\otimes P_\omega) && \lambda> 0
\end{align}
Punting on the question of existence for a moment, compare \cref{contobjective} to the \textit{Schr\"odinger Bridge problem} (SBP) from entropic optimal transport: given two marginals $\nu\in\Delta A, \mu\in\Delta \Omega$, the SBP is to find a joint coupling $P\in \Delta(A\times\Omega)$ which maximizes the objective
\begin{align}\label{eq:SBP}
    \sup_{P_\alpha=\nu,\ P_\omega=\mu} \int \frac{u(\alpha,\omega)}{\lambda} \ dP-D_{KL}(P\|P_\alpha\otimes P_\omega) && \lambda> 0
\end{align}
A coupling $P$ which attains the supremum in \cref{eq:SBP} is called a \textit{Schr\"odinger Bridge} between $\nu$ and $\mu$.  \\

The rational inattention objective in \cref{contobjective} looks very similar to the SBP in \cref{eq:SBP}. The \textbf{only} difference is that the SBP imposes a second marginal constraint $P_\alpha=\nu$. Given the resurgence in interest in the SBP, it seems eminently natural to consider looking at the rational inattention objective as a two-step problem with a nested SBP:
\begin{align}\label{eq:SBPapproach}
    \sup_{\nu\in\Delta A} V(\nu) && V(\nu):=\sup_{P\in\Pi(\nu,\mu)} \int \frac{u(\alpha,\omega)}{\lambda} \ dP-D_{KL}(P\|\nu\otimes \mu) 
\end{align}
where
\begin{align*}
    \Pi(\nu,\mu):=\bigg\{P\in \Delta(A\times\Omega): P_\alpha=\nu,\ P_\omega=\mu\bigg\}
\end{align*}
is standard notation in optimal transport for the set of probability measures over $A\times\Omega$ for which $P_\alpha=\nu$ and $P_\omega=\mu$. \\

\noindent 
\textit{Contribution.} The Schr\"odinger bridgehead approach, proposed in \cref{eq:SBPapproach}, separates the generalized rational inattention problem defined by \cref{contobjective} into an inner matching problem and an outer marginal-choice problem. The inner problem finds the Schr\"odinger Bridge between $\nu$ and $\mu$ for a fixed $\nu$, which we call the \textit{bridgehead}, while the outer problem finds the optimal bridgehead. The main contribution of this paper is to extend the rational inattention problem to general action spaces using entropic optimal transport. This is not to say that continuous action spaces are actually econometrically practical given finitely many observations; rather, this shows that it is valid to model an agent with a continuous action space as an agent with a finely discretized action space. \\

The economic content of the Schr\"odinger bridgehead approach is not obvious, but the mathematical motivation is clear. Since \textcite[hereafter, MM]{matvejka2015rational}, we have known that the \textit{optimal} information policy is an additive multinomial logit
\begin{align}\label{eq:MMMNL}
    P(\alpha|\omega)=\frac{e^{[u(\alpha,\omega)/\lambda]+\log P(\alpha)}}{Z(\omega;P_\alpha)} &&\text{where }Z(\omega;P_\alpha)=\sum_{\alpha \in A} e^{[u(\alpha,\omega)/\lambda]+\log P(\alpha)}
\end{align}
and that the optimal marginal $P_\alpha$ solves the Kuhn-Tucker condition
\begin{align}\label{eq:MMFOC}
    \int \frac{e^{u(\alpha,\omega)/\lambda}}{Z(\omega;P_\alpha)}\  d\mu(\omega)\leq 1
\end{align}
with equality almost everywhere. The struggle was always with the fact that \cref{eq:MMMNL} only provides a relationship between the \textit{optimal} conditional choice probability $P(\alpha|\omega)$ and the \textit{optimal} marginal $P_\alpha$, but the optimal marginal is somewhat mysterious: it typically lacks a closed form and is generically non-unique. \textcite[hereafter, CDL]{caplin2019rational} made significant progress in characterizing the supports of the optimal marginals, but hitherto thou shalt come. \\

In sum, we know that the challenging object in rational inattention is the marginal $P_\alpha$---once the optimal marginals are known, the optimal conditional choice probabilities $P(\alpha|\omega)$ follow from \cref{eq:MMMNL}. Thus, even if it is \textit{prima facie} economically nonsensical to ask what the optimal conditional choice probability is when the marginal is suboptimal---which is exactly what the SBP does---the marginal is nevertheless naturally an object of mathematical interest. The fact that the tools we use to study this object also allow us to extend rational inattention to continuous action spaces is an added bonus.\\

The main economic insight made clear by the Schr\"odinger bridgehead approach is that the  Lagrange multipliers (called the \textit{Schr\"odinger potentials}, and which can be infinite-dimensional) associated with the marginal constraints of the SBP are economically meaningful objects in their own right. The inner optimal transport problem emphasizes that the optimal conditional choice probability \textit{has to} eliminate residual matching surplus subject to the entropic cost. Once the bridge exhausts the matching gains, the payoff becomes additively separated into the expected value of state and action potentials, analogous to the Kantorovich potentials in vanilla optimal transport. Through additive separability, we can deduce that the action potential encodes the direction of greatest gain in $\Delta A$. The outer problem then highlights that these gains must \textit{also} be exhausted.   \\ 

More specifically, the action potential $a_\nu:A\to\R$ associated with the marginal $\nu$ describes the marginal benefit of re-assigning additional mass at the point $\alpha$ to $\nu$ from the perspective of the outer problem.\footnote{i.e. assuming that for a sub-optimal $\nu$, the agent nevertheless chooses the optimal joint---the Schr\"odinger bridge} Because the outer problem allows the agent to freely adjust $\nu$, the potential associated with the \textit{optimal} marginal must be zero on the consideration set and non-positive off support.\footnote{In continuous cases, this is actually a subtle point: we are accustomed in measure theory to treating things that are a.s. equal as equivalent, but such a custom may be unfounded in optimization contexts.} This provides a new interpretation for the necessary and sufficient first-order conditions outlined by CDL which characterize the optimal marginal.\\

On the other hand, the state potential $b_\nu:\Omega \to \R$ can encode the optimal conditional choice probability under the optimal marginal $\nu$. This is because the payoff is additively separated into the expected value of the state potential plus the action potential; but because the action potential is zero almost everywhere, we are left \textit{only} with the state potential, which coincides with the log of the normalizing factor in the conditional choice probability, seen in the denominator of \cref{eq:MMMNL}. The normalizing factor is called a \textit{partition function} in statistical mechanics. Just as the log of a partition function characterizes its statistical mechanical ensemble via its cumulant generating properties, $b_\nu$ characterizes properties of rationally inattentive choice. \\

Finally, we use the potentials to derive a connection between a computational tool for optimization in rational inattention---the Blahut-Arimoto algorithm---with the computational tool which re-ignited interest in entropic optimal transport---the Sinkhorn matrix-scaling algorithm. We show how the esoteric Blahut-Arimoto re-weighting process can be interpreted as an augmented matrix-scaling algorithm which improves the guess of the optimal marginal $\nu\in\Delta A$ after each step using the potential $a_\nu(\alpha)$ to guide the marginal in a more favourable direction. \\

\noindent 
\textit{Related literature.} Optimal transport has taken on an evangelical following in economics. This paper contributes to a growing literature of economic problems which \textcite{galichon2021unreasonable} refers to as \textit{optimal transport problems in disguise}. The Schr\"odinger Bridge problem was first studied by \textcite{schrodinger1931umkehrung}, who considered the most likely stochastic process connecting an initial distribution to a final distribution; however, its revival in the modern literature can be traced to  \textcite{cuturi2013sinkhorn}, who showed that computation of optimal transport problems could be improved by orders of magnitude by ``smoothing'' the problem using entropy-regularization. Much of the focus has centered on this use case, and particularly on asymptotic behaviour as the regularization term nears zero. To the best of our knowledge, computational efficiency has also been the only use case for entropic optimal transport in information economics \parencite{justiniano2025entropy}. This paper presents a non-asymptotic application: we cast the Shannon rational inattention model as an optimal Schr\"odinger bridgehead problem, and using results from entropic optimal transport, we reproduce and generalize the results of \textcite{matvejka2015rational,caplin2019rational}.

\section{The Schr\"odinger Bridge Problem}

We start with four important theorems in the theory of entropic optimal transport. Detailed proofs can be found in standard texts on the subject, such as \textcite{nutz2021introduction}.\footnote{One should warn for the sake of the interested reader that optimal transport problems are typically written with the opposite sign. $\sup$'s become $\inf$'s, utilities become costs, positives become negatives, \textit{etc.} }
\begin{theorem}
    For each $\nu$, there exists a unique Schr\"odinger bridge $P\in \Pi(\nu,\mu)$ that attains the supremum in \cref{eq:SBP}.  
\end{theorem}

\begin{theorem}[Existence and uniqueness of potentials]\label{propertiesofEOT}
    Let $P$ be the Schr\"odinger bridge that solves \cref{eq:SBP}. There exist functions $a_\nu\in L^1(\nu), b_\nu\in L^1(\mu)$ such that
    \begin{align}\label{structure}
        \frac{dP}{d(\nu\otimes\mu)} =e^{\frac{u(\alpha,\omega)}{\lambda}-a_\nu(\alpha)-b_\nu(\omega)}
    \end{align}
    These are called the \textbf{Schr\"odinger potentials}. Conversely, if there exist potentials such that \cref{structure} holds, then $P$ must solve \cref{eq:SBP}. Moreover, potentials are unique up to translation.\footnote{That is, for any two pairs of potentials $(a_\nu,b_\nu)$ and $(a_\nu',b_\nu')$, we have $a_\nu-a'_\nu=b_\nu'-b_\nu$.  }
\end{theorem}
\begin{corollary}[Additive separability]\label{cor:addsep}
    $\dis 
    V(\nu)=\int a_\nu(\alpha)\ d\nu(\alpha)+\int b_\nu(\omega)\ d\mu(\omega)$
\end{corollary}

\begin{theorem}[Schr\"odinger Equations]\label{Schr\"odingereq}
        $a_\nu$ and $b_\nu$ constitute a pair of Schr\"odinger potentials if and only if he satisfy
    \begin{align*}
        e^{a_\nu(\alpha)}= \int e^{\frac{u(\alpha,\omega)}{\lambda}-b_\nu(\omega)}\:d\mu(\omega) && e^{b_\nu(\omega)}=\int e^{\frac{u(\alpha,\omega)}{\lambda}-a_\nu(\alpha)}\: d\nu(\alpha)
    \end{align*}
\end{theorem}

\begin{theorem}[Duality] $\dis 
        V(\nu)=\inf_{a\in L^1(\nu), \ b\in L^1(\mu)}\int a\:d\nu + \int b\:d\mu +\iint e^{\frac{u}{\lambda}-a-b}\:d\nu\:d\mu - 1$
\end{theorem}

Some care must be taken when considering the distinction between a function $g$ and its associated equivalence class of functions $h$ for which $g=h$ a.e. When considering the optimality of the choice $P$ in a rational inattention problem, one must not only demonstrate optimality of choices on the support of $P$, but also that choices on the null set of $P$ are not strictly better than the choices in the support. As such, we need to proceed with caution; for instance, satisfying a necessary condition ``almost everywhere'' may not be good enough. An associated problem is that `choice' of $P$ does not precisely pin down its disintegration kernels or associated densities. Though not fatal, the ambiguity of possibly infinite densities clutters analysis with qualifiers. To simplify things, we restrict the agent to choosing versions of disintegrations and densities which satisfy
\begin{align}\label{densityregularity}
    \frac{dP(\alpha,\omega)}{d(P_\alpha\otimes \mu )(\alpha,\omega)}=\frac{dP(\alpha|\omega)}{dP(\alpha)}=\frac{dP(\omega|\alpha)}{d\mu(\omega)}<\infty && \frac{dP(\omega)}{d\mu(\omega)}=1
\end{align}
\textit{everywhere}, not just almost everywhere. This is without loss, since the definition of objective forces different versions to yield the same payoffs.

\section{Rational Inattention in Polish Spaces}

\subsection{Some Basic Existence Results}

The first question to ask is whether there even exists an information policy which maximizes the objective \cref{contobjective}. As usual, compactness of $A$ and upper semicontinuity of $u$ suffices.\\

Define
\begin{align*}
        f(\nu)=\int \log\pare{\int e^{u(\alpha,\omega)/\lambda}\ d\nu(\alpha)}\ d\mu(\omega)
\end{align*}
An important result in MM is that maximizing the discrete version of $f$ in tantamount to maximizing the discrete version of $V$. We will generalize this result in this paper. For now, let us establish the circumstances under which $f$ is maximized.
\begin{lemma}
    If $u(\cdot, \omega)$ is upper semicontinuous $\mu$-a.s. then $f$ is upper semicontinuous.
\end{lemma}
\begin{proof}
    $u\mapsto e^{u/\lambda}$ is continuous and non-decreasing, and thus $e^{u(\alpha,\omega)/\lambda}$ is upper semicontinuous, in addition to being lower bounded by $0$ and upper bounded by $e^{\overline u/\lambda}$, by assumption. By the Portmanteau theorem, if $\nu_n\to \nu$ weakly then
    \begin{align*}
        &\limsup_{n\to\infty}\int e^{u(\alpha,\omega)/\lambda}\ d\nu_n\leq \int e^{u(\alpha,\omega)/\lambda}\ d\nu\ ; \mu\text{-a.s.}
    \end{align*}
    By Fatou's lemma,
    \begin{align*}
        &\limsup_{n\to\infty} f(\nu_n) =\limsup_{n\to\infty}\int \log\pare{\int e^{u(\alpha,\omega)/\lambda}\ d\nu_n}\ d\mu(\omega)\leq\int \log\pare{\limsup_{n\to\infty}\int e^{u(\alpha,\omega)/\lambda}\ d\nu_n}\ d\mu(\omega) \leq f(\nu)
    \end{align*}
    
\end{proof}

\begin{theorem}
    If $A$ is compact and $f$ is upper semicontinuous, then $f$ attains its supremum in $\Delta A$.
\end{theorem}

\begin{proof}
    If $A$ is compact, then the space of probability measures over $A$ is compact in the weak topology. An upper semicontinuous function over a compact set is maximized.
\end{proof}

In the succeeding sections, we prove many theorems which are conditional on $V$ or $f$ attaining its maximum. The preceding results are sufficient, but not necessary.

\subsection{Optimal Information Policies}
With the entropic optimal transport tools in place, we can reproduce and extend the foundational results in the rational inattention literature. The reader will note that some of these results are almost trivial implications of the properties of the Schr\"odinger Bridge.\\

However, we begin with a peculiar result: if we write the objective as a single integral, then under the optimal information policy, the integrand is equal to a constant everywhere except on a null set, \textit{where it is less than the constant}. This unusual version of Gibbs' variational principle highlights the measure-theoretic caution that is sometimes needed in continuous rational inattention.

\begin{definition}
    A function $g:X\to \R$ is a plateau with respect to the measure $m$ if
    \begin{align*}
        m(\{x\in X:g(x)=\text{ess sup } g\})=m(X)
    \end{align*}
    and
    \begin{align*}
        g(x)\leq \text{ess sup }g && \forall \ x \in X
    \end{align*}
    That is, $g$ is equal to some constant $C$, $m$-a.e., and is less than $C$ everywhere.
\end{definition}

\begin{theorem}\label{gibbsproperty}
    If $P$ maximizes the objective \cref{contobjective} then for $\mu$-a.e. $\omega$,
    \begin{align*}
        \alpha\mapsto \frac{u(\alpha,\omega)}{\lambda}-\log\pare{\frac{dP}{d(P_\alpha\otimes \mu)}(\alpha,\omega)}
    \end{align*}
    is a $P(\cdot|\omega)$-plateau, $\mu$-a.s.
\end{theorem}
The proof of this theorem, given its variational nature, will be deferred to the appendix. For now, we point to several immediate corollaries which follow from the theorem.

\begin{corollary}\label{cor:logitisoptimal}
    If $P$ maximizes \cref{contobjective}, then
    \begin{align}\label{eq:logit}
        \frac{dP}{d(P_\alpha\otimes\mu)}(\alpha,\omega) = \frac{e^{u(\alpha,\omega)/\lambda}}{\int e^{u(\alpha',\omega)/\lambda}\ dP_\alpha(\alpha')}
    \end{align}
    $P$-almost surely.
\end{corollary}
\begin{corollary}
    The Schr\"odinger potential $a_\nu$ of an optimal marginal $\nu\in \argmax V$ is constant $\nu$-a.s.
\end{corollary}

\noindent Define the \textit{Jensen envelope} of $V$ to be
\begin{align*}
        f(\nu)=\int \log\pare{\int e^{u(\alpha,\omega)/\lambda}\ d\nu(\alpha)}\ d\mu(\omega)
\end{align*}
The following corollary and proof makes the name of $f$ self-explanatory.
\begin{corollary}\label{cor:fequalsV}
    $f\geq V$, and if $\nu$ maximizes $V$, then $f(\nu)=V(\nu)$. 
\end{corollary}

\begin{proof}
    The inequality follows from re-writing the SBP as an integral
    \begin{align*}
        \sup_{P\in\Pi(\nu,\mu)}\iint \log\pare{e^{u(\alpha,\omega)/\lambda} \pare{\frac{d\nu(\alpha)}{dP(\alpha|\omega)}} }\ dP(\alpha|\omega)\ d\mu(\omega)
    \end{align*}
    and applying Jensen's inequality and a change of measure. Note that finite KL-divergence implicitly encodes absolute continuity. If $\nu$ maximizes $V$, then the Schr\"odinger bridge from $\nu$ to $\mu$ maximizes the objective \cref{contobjective}. Hence the integrand is constant and we have an equality.
\end{proof}

\begin{lemma}[Properties of $f$]\label{lem:propertiesoff}
    $f$ has the following properties:
    \begin{enumerate}
        \item Let $\psi,\nu \in \Delta A$. The Gateaux derivative is
        \begin{align*}
            \dd_\psi f(\nu)=\int \frac{\int e^{u(\alpha,\omega)/\lambda}\ d\psi(\alpha) }{\int e^{u(\alpha,\omega)/\lambda}\ d\nu(\alpha)}-1\ d\mu(\omega)
        \end{align*}
        \item If $P$ is an optimal information policy and $P_\alpha=\nu$, then
        \begin{align}\label{eq:Voptimal}
            V(\nu)=\int \log\pare{\int e^{u(\alpha,\omega)/\lambda}\ d\nu(\alpha)}\ d\mu(\omega)
        \end{align}
        The Schr\"odinger potentials are $a_{\nu}(\alpha)=0$ and 
        \begin{align*}
        b_{\nu}(\omega)=\log\pare{\int e^{u(\alpha,\omega)/\lambda}\ d\nu(\alpha)}=\log Z(\omega;\nu)
        \end{align*}
        where $Z$ is the partition function (the denominator of \cref{eq:logit}). Thus,
        \begin{align*}
            f(\nu)=\int b_\nu(\omega)\ d\mu(\omega)=\int \log Z(\omega;\nu)\ d\mu(\omega)
        \end{align*}
        Moreover, $\dis\omega\mapsto   \frac{ e^{u(\alpha,\omega)/\lambda} }{\int e^{u(\alpha',\omega)/\lambda}\ d\nu(\alpha')}$ is in $L^1(\mu)$ for $\nu$-a.e. $\alpha$
    \end{enumerate}
\end{lemma}
\begin{proof}
    The Gateaux derivative follows from direct computation. \textcite[Theorem 2.27]{folland1999real} requires that
    \begin{align*}
        \int \frac{\int e^{u(\alpha,\omega)/\lambda}\ d\psi(\alpha) }{\int e^{u(\alpha,\omega)/\lambda}\ d\nu(\alpha)}\ d\mu(\omega)<\infty
    \end{align*}
    An inspection of the theorem's proof reveals that this is necessary to invoke the dominated convergence theorem. However, because the integrand above is non-negative, we can invoke Fatou's lemma instead to show that
    \begin{align*}
        \liminf_{\varepsilon\to 0} \frac{f(\nu_\varepsilon)-f(\nu)}{\varepsilon}\geq \int \frac{\int e^{u(\alpha,\omega)/\lambda}\ d\psi(\alpha) }{\int e^{u(\alpha,\omega)/\lambda}\ d\nu(\alpha)}-1\ d\mu(\omega)
    \end{align*}
    and hence $\dd_{\psi} f(\nu)=\infty$ whenever the $\dis 
        \int \frac{\int e^{u(\alpha,\omega)/\lambda}\ d\psi(\alpha) }{\int e^{u(\alpha,\omega)/\lambda}\ d\nu(\alpha)}\ d\mu(\omega)=\infty$.\\

    Suppose $P$ is an optimal information policy and $P_\alpha =\nu$. By \cref{cor:fequalsV}, $f(\nu)=V(\nu)$, which implies \cref{eq:Voptimal}. By \cref{cor:addsep}, $a_\nu(\alpha)=0$ and $b_\nu(\omega)=\log Z(\omega;\nu)$. It follows that
    \begin{align*}
         \frac{ e^{u(\alpha,\omega)/\lambda} }{\int e^{u(\alpha',\omega)/\lambda}\ d\nu(\alpha')}=e^{\frac{u(\alpha,\omega)}{\lambda}-b_\nu(\omega)}
    \end{align*}
    and by the Schr\"odinger equations,
    \begin{align*}
        \int e^{\frac{u(\alpha,\omega)}{\lambda}-b_\nu(\omega)}\ d\mu(\omega)=e^{a_\nu(\alpha)}<\infty
    \end{align*} 
    $\nu$-almost surely.
\end{proof}

This is our first hint at the intimate connections between $f$, $V$, $Z$, and $b$ which characterize the optimal information policy $P$.

\begin{theorem}\label{th:thennumaximizesV}
    If $\nu$ maximizes $f$, then $\nu$ maximizes $V$ and    \begin{align*}
        \frac{dP}{d(\nu\otimes\mu)}(\alpha,\omega) = \frac{e^{u(\alpha,\omega)/\lambda}}{\int e^{u(\alpha',\omega)/\lambda}\ d\nu(\alpha')}
    \end{align*}
    is a choice probability which maximizes \cref{contobjective}. 
\end{theorem}

\begin{proof}
    Let $\nu$ maximize $f$. By \cref{lem:propertiesoff}, the Gateaux derivative in the direction of $\psi$ is
    \begin{align*}
        \dd_{\psi} f(\nu)=\int \frac{\int e^{u(\alpha,\omega)/\lambda}d(\psi-\nu)(\alpha)}{\int e^{u(\alpha',\omega)/\lambda}d\nu(\alpha')}\ d\mu(\omega)
    \end{align*}
    whenever the integral is $<\infty$. \\
    
    Consequently, a necessary first order condition is that
    \begin{align*}
        \int \frac{ e^{u(\alpha,\omega)/\lambda}}{\int e^{u(\alpha',\omega)/\lambda}d\nu(\alpha')}\ d\mu(\omega)=1
    \end{align*}
    for $\nu$-a.e. $\alpha$. Then,
    \begin{align*}
        a(\alpha)=0 && b(\omega)=\log\pare{\int e^{u(\alpha,\omega)/\lambda}\ d\nu(\alpha)}
    \end{align*}
   can be verified to be valid Schr\"odinger potentials by \cref{Schr\"odingereq}. By the duality formula,
    \begin{align*}
        V(\nu)=\int b(\omega)\ d\mu(\omega)=f(\nu)
    \end{align*}
    Now consider any other $\nu'$. Define $a'=0$ as a ``candidate potential'' and
    \begin{align*}
        b'(\omega)=\log\pare{\int e^{u(\alpha,\omega)/\lambda}\ d\nu'(\alpha)}
    \end{align*}
    via the Schr\"odinger equation as the other candidate potential, such that 
    \begin{align*}
        V(\nu')\leq \int a'\ d\nu' + \int b'\ d\mu+\iint e^{\frac{u}{\lambda}-a'-b'}\ d\nu'\ d\mu-1=\int b'\ d\mu
    \end{align*}
    The candidate potentials are not necessarily the true potentials associated with $\nu'$. The first inequality is strict if $b'$ is not the true potential of $\nu'$. The duality formula implies that
    \begin{align*}
        V(\nu')\leq \int b'(\omega)\ d\mu(\omega)=f(\nu')\leq f(\nu)=V(\nu)
    \end{align*}This proves that if $\nu$ maximizes $f$, then $\nu$ maximizes $V$. Via \cref{propertiesofEOT}, 
    \begin{align*}
        \frac{dP}{d(\nu\otimes\mu)}(\alpha,\omega) = \frac{e^{u(\alpha,\omega)/\lambda}}{\int e^{u(\alpha',\omega)/\lambda}\ d\nu(\alpha')}
    \end{align*}
    defines the Schr\"odinger bridge from $\nu$ to $\mu$.
\end{proof}

\begin{corollary}[\textcite{matvejka2015rational}, Lemma 2]
    Suppose $f$ has a maximum in $\Delta A$.  Then $\nu$ maximizes $V$ if and only if $\nu$ maximizes $f$.
\end{corollary}

\begin{proof}
    $(\Leftarrow)$ See \cref{th:thennumaximizesV}.\\

    \noindent $(\Rightarrow)$ if $\nu$ maximizes $V$, let $\nu'$ maximize $f$. By \cref{cor:fequalsV}, $f(\nu)=V(\nu)$. By \cref{th:thennumaximizesV}, $\nu'$ maximizes $V$. It follows that $f(\nu)=V(\nu)=V(\nu')=f(\nu')$. Hence $\nu$ maximizes $f$.
\end{proof}

This finalizes the connection between $V$ and $f$, demonstrating that to find an optimal information policy, it is necessary and sufficient to find a marginal in $\Delta A$ which maximizes $f$---a much cleaner function to maximize than $V$---and then defining the information policy via \cref{eq:logit}. The following theorem confirms that $f$ is concave, thus characterizing the set of optimal marginal probabilities as convex.

\begin{lemma}[\textcite{caplin2013behavioral}, Theorem 2]\label{lem:concavity}
    $f$ is concave, and if $\nu$ and $\nu'$ both maximize $f$, then 
    \begin{align*}
        \int e^{u(\alpha,\omega)/\lambda}\ d\nu(\alpha)=\int e^{u(\alpha,\omega)/\lambda}\ d\nu'(\alpha)
    \end{align*}
    for $\mu$-a.e. $\omega$.
\end{lemma}
\begin{proof}
    To see that $f$ is concave, simply observe that
    \begin{align*}
        f(\beta\nu+(1-\beta)\nu')&=\int \log\pare{\beta \int e^{u(\alpha,\omega)/\lambda}\ d\nu(\alpha)+(1-\beta)\int e^{u(\alpha,\omega)/\lambda}\ d\nu'(\alpha)} \ d\mu(\omega)\\
        &\geq\int \beta \log \pare{\int e^{u(\alpha,\omega)/\lambda}\ d\nu(\alpha)}+(1-\beta)\log\pare{\int e^{u(\alpha,\omega)/\lambda}\ d\nu'(\alpha)}\ d\mu(\omega)\\
        &=\beta f(\nu)+(1-\beta)f(\nu')
    \end{align*}
    with strict equality if
    \begin{align*}
        \int e^{u(\alpha,\omega)/\lambda}\ d\nu(\alpha)\neq \int e^{u(\alpha,\omega)/\lambda}\ d\nu'(\alpha)
    \end{align*}
    on a set of positive $\mu$.
\end{proof}

We arrive at the necessary and sufficient conditions for optimality of $\nu$, derived by CDL. 
\begin{theorem}[CDL, Proposition 1, First-Order Condition]\label{th:FOC}
    Suppose $f$ has a maximum in $\Delta A$. $\nu$ maximizes $f(\nu)$ if and only if for $\nu$-a.e. $\alpha$,
    \begin{align}\label{FOC}
        \int \frac{e^{u(\alpha ,\omega)/\lambda}}{\int e^{u(\alpha ',\omega)/\lambda} \ d\nu(\alpha')}\ d\mu(\omega)=1
    \end{align}
     and
    \begin{align}\label{FOCleq}
        \int \frac{e^{u(\alpha ,\omega)/\lambda}}{\int e^{u(\alpha ',\omega)/\lambda} \ d\nu(\alpha')}\ d\mu(\omega)\leq 1
    \end{align}
    for \textbf{every} $\alpha \in A$. In other words, $\dis \alpha\mapsto \int \frac{e^{u(\alpha ,\omega)/\lambda}}{\int e^{u(\alpha ',\omega)/\lambda} \ d\nu(\alpha')}\ d\mu(\omega)$ is a $\nu$-plateau.
\end{theorem}
This proof once again relies on variational arguments, so we defer it to the appendix. In the meantime, observe that by \cref{lem:concavity}, we have that if $\nu$ and $\nu'$ both maximize $V$, then 
\begin{align*}
        \int e^{u(\alpha,\omega)/\lambda}\ d\nu(\alpha)=\int e^{u(\alpha,\omega)/\lambda}\ d\nu'(\alpha)=z(\omega)
\end{align*}
$\mu$-a.s. This implies the following.

\begin{corollary}\label{cor:support}
    If $\nu$ maximizes $V$ and $f$, then for $\nu$-a.e. $\alpha$,
    \begin{align*}
    \int \frac{e^{u(\alpha,\omega)/\lambda}}{z(\omega)}\ d\mu(\omega)=1
\end{align*}
\end{corollary}

\Cref{lem:concavity,cor:support} jointly characterize the argmax of $V$ and $f$. \Cref{lem:concavity} shows that the argmax is convex. \Cref{cor:support} shows that each $\nu\in\argmax f$ must have a support which is in the closure of
\begin{align*}
    S=\bigg\{\alpha \in A: \int \frac{e^{u(\alpha,\omega)/\lambda}}{z(\omega)}\ d\mu(\omega)=1\bigg\}
\end{align*}
That is, consideration sets are restricted to being subsets of $\overline S$.
\begin{definition}
    The \textit{consideration set} associated with $\nu$ is the support of $\nu$.
\end{definition}

The final two results concern the \textit{invariant likelihood ratios}, which are used in the posterior-based approach which is the subject of CDL.

\begin{theorem}[CDL, Proposition 2, Invariant Likelihood Ratio]
    Let $P$ maximize \cref{contobjective} such that $a(\alpha)=0$, $b(\omega)=\log Z(\omega;P_\alpha)$ are the associated potentials. For $\nu$-a.e. $\alpha$,
    \begin{align*}
        dP(\omega|\alpha)=\frac{e^{u(\alpha,\omega)/\lambda}}{Z(\omega;P_\alpha)}\ d\mu(\omega)
    \end{align*}
\end{theorem}

\begin{proof}
    This follows from \cref{propertiesofEOT}.
\end{proof}

\begin{corollary}[ILR Inequalities]
    Let $P$ maximize \cref{contobjective} and $P_\alpha=\nu$. For every $\alpha$, for $\nu$-a.e. $\alpha'$,
    \begin{align*}
        \int \exp\pare{\frac{u(\alpha,\omega)-u(\alpha',\omega)}{\lambda}}\ dP(\omega|\alpha')\leq 1
    \end{align*}
    For $\nu$-a.e. $\alpha$, for $\nu$-a.e. $\alpha'$, 
    \begin{align*}
        \int \exp\pare{\frac{u(\alpha,\omega)-u(\alpha',\omega)}{\lambda}}\ dP(\omega|\alpha)= 1
    \end{align*}
\end{corollary}
\begin{proof}
    Fix $\alpha \in A$. By \cref{th:FOC}, for $\nu$-a.e. $\alpha'\in A$,
    \begin{align*}
        \int \exp\pare{\frac{u(\alpha,\omega)-u(\alpha',\omega)}{\lambda}}\ dP(\omega|\alpha')&=\int \exp\pare{\frac{u(\alpha,\omega)-u(\alpha',\omega)}{\lambda}}\ \frac{e^{u(\alpha',\omega)/\lambda}}{Z(\omega;P_\alpha)}\ d\mu(\omega)\\
        &=\int \frac{e^{u(\alpha,\omega)/\lambda}}{Z(\omega;P_\alpha)}\ d\mu(\omega)\leq 1
    \end{align*}
    with equality for $\nu$-a.e. $\alpha$.
\end{proof}

CDL use the invariant likelihood ratios (ILR) to construct a belief-based characterization of consideration sets. 
For any $\overline \alpha \in B\subset A$ and any posterior belief $P(\omega|\overline \alpha)$, one can use invariant likelihood ratios to define the remaining posteriors for signals in $B$ by
\begin{align*}
    P(S|\alpha)=\int_S\exp\pare{\frac{u(\alpha,\omega)-u(\overline \alpha,\omega)}{\lambda}}\ dP(\omega|\overline \alpha)
\end{align*}
A closed set $B$ is then a consideration set, given $P(\omega|\overline\alpha)$, only if there exists $\nu\in\Delta A$ such that $\text{supp }\nu=B$ and
\begin{align*}
    \mu(S)=\iint _S\exp\pare{\frac{u(\alpha,\omega)-u(\overline \alpha,\omega)}{\lambda}}\ dP(\omega|\overline \alpha)\ d\nu(\alpha)\equiv \int P(S|\alpha)\ d\nu(\alpha)
\end{align*}
for every measurable $S\subset \Omega$. That is, it is possible to weight the posteriors such that it averages to the prior. They state this as $\mu$ lies in the interior of the convex hull of 
\begin{align*}
    \bigg\{m\in \Delta \Omega: dm=\exp\pare{\frac{u(\alpha,\omega)-u(\overline \alpha,\omega)}{\lambda}}\ dP(\omega|\overline \alpha)\text{ for some }\alpha \in B\bigg\}
\end{align*}

In essence, given any consideration set $B$ or cardinality $k$, for any $\overline\alpha\in B$ the other actions in the consideration set provide $k-1$ linear restrictions on the posterior $P(\omega|\overline \alpha)$. If the $k-1$ restrictions are linearly independent and the simplex $\Delta \Omega$ is $(k-1)$-dimensional, then for any $\overline \alpha\in B$, this pins down a unique posterior $P(\omega|\overline \alpha)$. This posterior, in turn, pins down the set of all posteriors $\{P(\omega|\alpha):\alpha \in B\}$ via the ILRs.

\section{Schr\"odinger Potentials}

The most economically useful objects introduced by the Schr\"odinger Bridge approach are the Schr\"odinger potentials, which are the infinite-dimensional analogs of Lagrange multipliers in entropic optimal transport. We highlight several ways in which the potentials can help us characterize and better understand the optimal information policy.

\subsection{Potentials Under Optimal Information Policies}

The role and nature of the potentials under the optimal information policy $P$ with optimal marginal $P_\alpha=\nu$ is given by \cref{cor:addsep,lem:propertiesoff}. To recap, $V(\nu)$ is additively separable, even if $\nu$ is not optimal,
\begin{align*}
    V(\nu)=\int a_\nu \ d\nu + \int b_\nu \ d\mu
\end{align*}
and when $\nu$ is optimal, $a_\nu$ is constant $\nu$-a.s. Our convention is to normalize $\dis \int a_\nu\ d\nu=0$.\\

Schr\"odinger duality is analogous to the familiar \textit{Kantorovich duality} from vanilla optimal transport. The additive separability of the objective function when evaluated at a Schr\"odinger bridge implies that the bridge exhausts all gains that can be made by matching states with actions, subject to the information cost. We are left with two terms, which summarize the direction in which the agent would like to change $\nu$ and $\mu$, respectively. The agent cannot choose $\mu$, but he can choose the optimal bridgehead $\nu$, and therefore, if the agent is maximizing the objective, then the associated Schr\"odinger potential should be constant---and without loss, zero.\\

This direct interpretation of the Schr\"odinger potentials becomes easy to see in the finite case---highlighting that the economic relevance of the potentials is certainly not limited to the continuous case. Here, a simple application of the envelope theorem reveals that the Schr\"odinger potential at a given point is the infinitessimal change in the objective induced by shifting an infinitessimal amount of probability mass away from all other points and towards that given point---i.e., the directional derivative in the direction of the Dirac measure at that point.

\begin{theorem}\label{th:potentials}
    If $A$ is finite and $\nu(\alpha)>0$ for all $\alpha \in A$, then the change in $V$ from re-assigning mass to the point $\alpha^* \in A$ is
    \begin{align*}
        \dd_{\delta_{\alpha^*}}V(\nu)=a_\nu(\alpha)-E_\nu[a_\nu]
    \end{align*}
    $\nu$-a.s., where $\delta_{\alpha^*}$ is the Dirac probability mass at $\alpha^*$. This is the Gateaux derivative of $V$ in the direction of $\delta_{\alpha^*}$. Similarly, if $\Omega$ is finite, then the Gateaux derivative of $V$ as $\mu$ is shifted in the direction of $\delta_{\omega^*}$ is
    \begin{align*}
        \dd_{\delta_{\omega^*}} V(\nu)=b_\nu(\omega^*)-E_\mu[b_\nu]
    \end{align*}
\end{theorem}
\begin{proof}
    First, let $A$ be finite. Applying the Envelope theorem to the duality formula, we obtain
    \begin{align*}
        \dd_{\delta_{\alpha^*}}V(\nu)&=\sum_{\alpha'\in A}\pare{a_\nu(\alpha)+\int e^{\frac{u(\alpha,\omega)}{\lambda}-a_\nu(\alpha)-b_\nu(\omega)}\ d\mu(\omega)}[\delta_{\alpha^*}(\alpha)-\nu(\alpha)] \\&=\sum_{\alpha'\in A}\pare{a_\nu(\alpha)+1}[\delta_{\alpha^*}(\alpha)-\nu(\alpha)] \\
        &=a_\nu(\alpha^*)-E_\nu[a_\nu]
    \end{align*}
    from whence the first result follows. The second result follows from a similar derivation.
\end{proof}

\subsection{Statistical Mechanical Connections}

Although not of paramount economic importance, the obvious connections to statistical mechanics are worth pausing on and outlining in some detail. The fact that the action potential is constant implies that if $\nu$ maximizes $V$ we get (after normalizing $E_\nu[a_\nu]=0$)
\begin{align*}
    b_\nu(\omega)=\log Z(\omega;\nu)=\log\pare{\int e^{u(\alpha,\omega)/\lambda}\ d\nu(\alpha)}
\end{align*}
In statistical mechanics, the log of the partition function is an important object which encodes the moments of the system. The same principle applies here, but in an economic context.
\begin{theorem}[Cumulant Generation] Let $P$ maximize \cref{contobjective} such that $a(\alpha)=0$, $b(\omega)=\log Z(\omega;P_\alpha)$ are the associated potentials. The $\omega$-conditional moments of utility are given by
\begin{align*}
    E_P[u(\alpha,\omega)|\omega]=\frac{\dd b(\omega)}{\dd (1/\lambda)} && \Var(u(\alpha,\omega)|\omega)=\frac{\dd^2 b(\omega)}{\dd (1/\lambda)^2}
\end{align*}
The information gain from $P(\alpha|\omega)$ over $P(\alpha)$ in state $\omega$ is
\begin{align*}
    D_{KL}(P(\alpha|\omega)\|P(\alpha))=E_P\brac{\log\pare{\frac{dP(\alpha|\omega)}{dP(\alpha)}}\ | \ \omega}=b(\omega)-\frac1\lambda \frac{\dd b(\omega)}{\dd (1/\lambda)}=\frac{\dd [\lambda b(\omega)]}{\dd \lambda}
\end{align*}
\end{theorem}
The proofs are straightforward and left as an exercise to the reader. We dive straight into the interpretation: recall from \cref{th:potentials} that $b(\omega)$ is the willigness-to-pay, in $\lambda$-normalized utils, to marginally increase the probability of the state $\omega$. The first moment condition states that the rationally inattentive agent's willingness-to-pay changes as the information cost goes down, and it changes commensurate to the expected utility in state $\omega$, given the optimal information policy $P$. We can think of $\lambda$ as a conversion factor between utils and bits of information. Accordingly, the second moment condition states that the agent's willingness-to-pay (in bits) for a marginal increase in the probability of state $\omega$ goes down as the exchange rate between bits and utils worsens.\\

From the preceding theorem follows the free energy minimization interpretation of the rational inattention problem. The analogous thermodynamic objects are as follows: $\lambda$ is the temperature, $-u$ is the energy, the information gain is the entropy, and the objective \cref{contobjective} minimizes the average free energy.
\begin{corollary}
    Define the relative entropy in state $\omega$ under $Q$ by
    \begin{align*}
        S_Q(\omega)=D_{KL}(Q(\alpha|\omega)\|Q(\alpha))=E_{\alpha \sim Q_\alpha}\brac{\pare{\frac{dQ(\alpha|\omega)}{dQ(\alpha)}}\log \pare{\frac{dQ(\alpha|\omega)}{dQ(\alpha)}}}
    \end{align*}
     and define $H_Q(\omega)=-E_Q[u(\alpha,\omega)|\omega]$ as the negative of the $\omega$-conditional expected utility under $Q$. Then, 
    \begin{align*}
        S_P(\omega)=\frac{H_P(\omega)}{\lambda}+\log Z(\omega;P_\alpha)\equiv \frac{\dd [\lambda \log Z(\omega;P_\alpha)] }{\dd \lambda}
    \end{align*}
    Finally, define the average free energy $G_Q(\omega)$ by
    \begin{align*}
        G_Q(\omega)= H_Q(\omega)-\lambda S_Q(\omega)
    \end{align*}
    Then, if $P$ maximizes the objective in \cref{contobjective} then $P$ minimizes average free energy.
\end{corollary}
Where our system uniquely departs is that we can define the thermodynamic quantities at each point $\omega \in \Omega$. A closer examination shows that it is a minor miracle that these thermodynamic connections appear at all---because it cannot be said from simply looking at the objective \cref{contobjective} that the agent chooses $P(\alpha|\omega)$ so as to minimize $G_P(\omega)$ in \textit{that} state specifically. The reason is that by changing $P(\alpha|\omega)$, one also changes $P(\alpha)$---a term which shows up in all other states. \\ 

The magic lies in the functional form. As seen in \cref{gibbsproperty}, a consequence of the log functional form of Shannon information
\begin{align*}
    \int \log\pare{\frac{dP}{d(P_\alpha\otimes \mu)}}\ dP
\end{align*}
is that even though $P$ appears in both the integrand and the integrator $dP$, the only first-order effect is through the integrator. That is, when analyzing local perturbations in $P$, we can ignore the effect on the integrand. This is an instance of the variational results commonly named for Gibbs, Donsker, and Varadhan.\\

Because the marginal $P_\alpha$ only shows up in the integrand, changing $P(\alpha|\omega)$ does not have a first-order affect on the mutual information through $P(\alpha)$. Locally, the agent can behave as though he is maximizing $u(\alpha,\omega)-D_{KL}(P(\alpha|\omega)\|P(\alpha))$ at every state $\omega$---or conversely, minimizing the free energy at every state.

\subsection{A Second Look at the Normalization Condition}
Recall the definition of $f$ in the discrete case
\begin{align*}
    f(\nu)=\int \log\pare{\sum_{\alpha \in A}e^{u(\alpha,\omega)/\lambda}\ \nu(\alpha) }\ d\mu(\omega)
\end{align*}
and the Kuhn-Tucker condition for its optimality, found in \cref{th:FOC}:
\begin{align}\label{eq:CDLKT}
    \int \frac{e^{u(\alpha,\omega)/\lambda}}{\sum_{\alpha'\in A} e^{u(\alpha'\omega)/\lambda}\ \nu(\alpha') }-1\ d\mu(\omega)\leq 0 && \text{everywhere, with equality $\nu$-a.e.}
\end{align}
As MM and CDL recognized early on, $f$ is an important object in rational inattention because it is a straightforward function of $\nu$ whose argmax coincides with the optimal marginals. However, $f$ lacks a clear economic interpretation. MM understood the equality in \cref{eq:CDLKT} as a ``normalization condition''; that is, when one plugs in the multinomial logit \cref{eq:MMMNL}, a necessary condition is that $\int P(\alpha|\omega)-1\ d\mu(\omega)$ must sum to $0$ a.s. The main result of CDL is that \cref{eq:CDLKT} is the Kuhn-Tucker condition for the maximization of $f$, and because $f$ is concave, Kuhn-Tucker is not only necessary, but also sufficient. CDL explain the Kuhn-Tucker conditions as a fixed-point of the Blahut-Arimoto algorithm (itself a complicated subject, which we address next), which increases $f$ with every iteration. \\

But this is about as far as we get to interpreting $f$ and its Kuhn-Tucker condition---as a fixed point and as a normalization condition. The Schr\"odinger potentials offer a second look at the esoteric conditions in \cref{th:FOC}. From hereon, let us specify the Schr\"odinger potential to be the functions which satisfy the Schr\"odinger equations exactly, so that: 
\begin{align*}
    a_\nu(\alpha)=\log\pare{\int e^{u(\alpha,\omega)-b(\omega)}\ d\mu(\omega)} 
\end{align*}
everywhere, not just $\nu$-a.e.\footnote{In the mathematical literature, the Schr\"odinger potentials are typically allowed to differ on null sets---in those cases, it usually makes no difference.} The main result to follow is a simplification of \cref{eq:CDLKT}.
\begin{theorem}
    Normalize $E_\nu[a_\nu]=0$. \Cref{eq:CDLKT} holds if and only if: 
    \begin{align}\label{eq:VKT}
    a_\nu(\alpha)\leq 0 && \text{everywhere, with equality $\nu$-a.e.}
\end{align}
i.e., $a_\nu$ is a plateau function.
\end{theorem}
\begin{proof}
    
 By \cref{th:FOC}, if \cref{eq:CDLKT} holds then $b_\nu$ is the log-partition function, and thus
\begin{align*}
    &\int \frac{e^{u(\alpha,\omega)/\lambda}}{\sum_{\alpha'\in A} e^{u(\alpha'\omega)/\lambda}\ \nu(\alpha') }-1\ d\mu(\omega)= \int e^{\frac{u(\alpha,\omega)}{\lambda}-b_\nu(\omega)}\ d\mu(\omega)-1 = e^{a_\nu(\alpha)}-1\\
    &\implies \text{sign}\pare{\int \frac{e^{u(\alpha,\omega)/\lambda}}{\sum_{\alpha'\in A} e^{u(\alpha'\omega)/\lambda}\ \nu(\alpha') }-1\ d\mu(\omega)}=\text{sign}(a_\nu(\alpha))
\end{align*}
which implies \cref{eq:VKT}. Conversely, if $a_\nu$, $b_\nu$ satisfy the Schr\"odinger equations and \cref{eq:VKT} holds, then $b_\nu$ is the log-partition function. The sign equation above holds and \cref{eq:CDLKT} is thus implied.
\end{proof}

Recall, by \cref{th:potentials}, $ \dd_{\delta _\alpha}V(\nu)=a_\nu(\alpha)$. In words, the preceding derivations belie a sequence of logical steps that stem from the basic results in entropic optimal transport: the Schr\"odinger bridge must exhaust all matching gains (subject to the entropic cost), for any fixed $\nu$; duality tells us that the value $V$ can be decomposed into action and state potentials; by the Envelope Theorem, the action potential $a_{\nu}(\alpha)$ encodes the marginal gain (in $V$) of re-assigning mass towards a given action $\alpha$. It follows that $\nu$ is optimal only if the marginal gain is zero almost everywhere on the consideration set and non-positive everywhere else. Thus we arrive at a simple interpretation of the Kuhn-Tucker condition of $f$ in \cref{th:FOC} as the condition which is satisfied when there are no possible gains from re-assigning mass to a given action $\alpha$, assuming matching gains are likewise exhausted.\\

The continuous case is more delicate because it is not a continuous maneuver to re-assign mass to a single point when points have zero mass, but the idea is the same: the necessary condition for $\nu$ to be optimal is that re-assigning mass to any $\alpha$ from other points in the support offers non-positive gains, and for $\nu$-a.e. $\alpha$, re-assigning mass to $\alpha$ offers exactly zero gains.

\subsection{A Second Look at the Blahut-Arimoto Algorithm}\label{sec:blahutarimoto}

The Blahut-Arimoto algorithm is a common method for finding a $\nu$ that comes close to maximizing $V$. It works by increasing $f$ after every iteration. But given that $f$ is already opaque, it can be challenging to imagine exactly why the algorithm works. CDL explain that the algorithm ``twists the state-dependent choice in the direction of the high payoff states [in a way that] ensures that these twists average out to one.'' But what exactly are these twists and what it would mean for an agent to actually implement this algorithm?\\

As we will show, the Blahut-Arimoto can be interpreted as twisting $\nu_n$ in the direction of $\exp\pare{a_{\nu_n}}$. Equipped with our understanding of what the action potential $a_\nu(\alpha)=\dd_{\delta_\alpha}V(\nu)$ represents---i.e., the marginal payoff from uniformly transporting mass towards the point $\alpha$---the motivation for the agent to `twist' $\nu$ in such a way becomes eminently clear. The Blahut-Arimoto algorithm can be seen as `gradient ascent' for $\nu$, i.e.,
\begin{align*}
    \log \nu_{n+1}(\alpha)=\log \nu_n(\alpha)+\dd_{\delta_\alpha} V(\nu_n)
\end{align*}

As an analogy, consider the work-life balance of a junior associate in a Wall Street firm. She has two actions: be at home with family or be at work, ready to address the firm's needs at a moment's notice. Every evening, she i) commits to spending a certain number of hours at home and ii) reads up on international markets to project when she will likely be needed at work, and commits to being at work during specific times of day. The following day, $\omega \in \Omega$ is realized, representing the times of the day when the labor of a junior associate is in highest demand. In this analogy, the notion of `exhausting matching surplus subject to information costs' represents optimally acquiring information and choosing when to be at work, conditional on her commitment to only work a certain number of hours. The notion of her action potential represents whether she would achieve higher or lower payoffs if she were to commit to greater or fewer hours of work overall. The Blahut-Arimoto algorithm corresponds roughly to iteratively increasing or decreasing the number of hours or work, conditional on optimally allocating her work hours based on the information she acquires. \\

To arrive at the Blahut-Arimoto algorithm from the Schr\"odinger bridge, we actually have begin our discussion elsewhere. Previous approaches to rational inattention emphasize a posterior-based approach; that is, finding posteriors that average to the prior. The discrete Schr\"odinger Bridge problem involves a similar problem of finding weights which average properly---albeit with two marginal constraints instead of one. In particular, it is an instance of the \textit{matrix scaling problem}, which seeks to find two sets of weights which, when organized into diagonal matrices, scales a given matrix to be row- and column-stochastic. \\

Formally, let $A$ and $\Omega$ be discrete. Given $\nu$, the Schr\"odinger bridge dual problem can be re-written as the problem of finding $a,b$ such that
\begin{align*}
    \sum_{\alpha \in A}\nu(\alpha)e^{u(\alpha,\omega)/\lambda}e^{-a(\alpha)}e^{-b(\omega)}=1\\
    \sum_{\omega\in \Omega}\mu(\omega)e^{u(\alpha,\omega)/\lambda}e^{-a(\alpha)}e^{-b(\omega)}=1
\end{align*}
almost surely. Recall that then,
\begin{align*}
    V(\nu)=\sum_{\alpha \in A}a(\alpha)\cdot \nu(\alpha)+\sum_{\omega\in \Omega} b(\omega)\cdot \mu(\omega)&& P(\alpha,\omega)=\nu(\alpha)\mu(\omega)e^{u(\alpha,\omega)/\lambda}e^{-a(\alpha)}e^{-b(\omega)}
\end{align*}
The high-level idea is that if the agent is efficiently matching states and characteristics subject to the cost, then the payoff $V$ should be separable into a utility shock from sub-optimal $\nu$ and a utility shock from sub-optimal $\mu$; that is, there should be no further matching surplus. The choice probability $P(\alpha,\omega)$ should be higher or lower than the baseline $\nu(\alpha)\mu(\omega)$ depending on the \textit{relative} utility of event $(\alpha,\omega)$---relative, that is, given the average utility shocks $a(\alpha)$ and $b(\omega)$. Finding a pair of potentials is, at face value, non-trivial. Fortunately, the Sinkhorn algorithm, which involves alternating between defining $a_n$ and defining $b_n$ using the equations above, converges to a unique solution, given $\nu$ and $\mu$. \\

The rational inattention problem has a similar flavour, but differs in one key respect: we are not given $\nu$, but we are given $a$, or at least we know exactly what the optimal potential $a$ \textit{should} be---namely, zero. So what if instead of being taking $\nu$ as given, we took $a(\alpha)=0$ as given? We would instead need to find $b$ and $\nu$ such that
\begin{align*}
    \sum_{\alpha \in A}\nu(\alpha)e^{u(\alpha,\omega)/\lambda}e^{-b(\omega)}=1\\
    \sum_{\omega\in \Omega}\mu(\omega)e^{u(\alpha,\omega)/\lambda}e^{-b(\omega)}\leq1 && \text{with equality on the support of }\nu
\end{align*}
In this setting, the agent would first find $b(\omega)$ which satisfies the second condition
\begin{align}\label{eq:eubinequality}
    \sum_{\omega\in \Omega}\mu(\omega)e^{u(\alpha,\omega)/\lambda}e^{-b(\omega)}\leq1 
\end{align}
with equality at some $\alpha$. For each $\alpha$ where \cref{eq:eubinequality} holds with equality, set
\begin{align*}
    P(\omega|\alpha)=\mu(\omega)e^{u(\alpha,\omega)/\lambda}e^{-b(\omega)}
\end{align*}
and finally, the agent finds weights such that the posteriors average to the prior:
\begin{align}\label{eq:weightstosumtoone}
    \frac{1}{\mu(\omega)}\sum_{\alpha \in A}\nu(\alpha)P(\omega|\alpha)=\sum_{\alpha \in A}\nu(\alpha)e^{u(\alpha,\omega)/\lambda}e^{-b(\omega)}=1
\end{align}
The idea here is that $\frac{u(\alpha,\omega)}{\lambda}-b(\omega)$ represents the utility of action $\alpha$ in state $\omega$ \textit{adjusting for the utility in state $\omega$ averaged over other options.} If the agent receives the recommendation $\alpha$, he infers that the state is more likely to be one where the action is ``good.'' But the agent cannot compare states simply by utility, because some states are inherently better than others. Instead, he infers that the action is good \textit{relative to the state}---that is, adjusting for $b(\omega)$. Indeed, we arrive at a invariant likelihood result which describes the likelihood ratio of two states given the same signal (instead of describing the likelihood of the same state $\omega$ given two different signals $\alpha,\alpha'$):
\begin{align*}
    \frac{P(\omega|\alpha)}{P(\omega'|\alpha)}=\exp\pare{\brac{\frac{u(\alpha,\omega)}{\lambda}-b(\omega)}-\brac{\frac{u(\alpha,\omega')}{\lambda}-b(\omega')}}
\end{align*}
Once the agent defines $P(\omega|\alpha)$, all he needs to do is then define a set of probability weights such that the posterior averages to the prior $\mu(\omega)$, as seen in \cref{eq:weightstosumtoone}.\\

The trouble, as usual, is that finding a coherent pair $(\nu,b)$ is a non-trivial task. The control function $-b(\omega)$ adjusts for the average utility shock in state $\omega$---but average with respect to what? Under the optimal information policy, the average is with respect to $\nu$---a thus far undefined set of weights for the posteriors $P(\omega|\alpha)$, which are to be defined by $b(\omega)$.\\

We can augment the Sinkhorn algorithm as follows. The agent starts with a probability $\nu_0(\alpha)$ supported over the entirety of $A$, and defines the candidate potential $b_0(\omega)$---the state-level utility shock---by
\begin{align*}
    b_0(\omega)=\log\pare{\sum_{\alpha\in A}\nu_0(\alpha) e^{u(\alpha,\omega)/\lambda}}
\end{align*}
The agent then defines the other candidate potential $a_0$ via the Schr\"odinger equations:
\begin{align*}
    a_0(\alpha)=\log\pare{\sum_{\omega\in \Omega}\mu(\omega)e^{u(\alpha,\omega)/\lambda}e^{-b_0(\omega)}}
\end{align*}
This signifies, after controlling for state-specific average utility shocks, the residual variation across actions that makes certain actions more or less preferable on average. If $\exp\pare{\frac{u(\alpha,\omega)}{\lambda}-b(\omega)}$ is consistently $>1$, then that suggests the agent should be shifting more probability weight towards $\alpha$. To wit, $\nu_1$ is set so that it updates $\nu_0$ in that direction:
\begin{align*}
    \nu_1(\alpha)=\nu_0(\alpha) e^{a_0(\alpha)}
\end{align*}
The entire algorithm can be summed up as follows:
\begin{align*}
    &b_n(\omega)=\log\pare{\sum_{\alpha \in A}\nu_n(\alpha)e^{u(\alpha,\omega)/\lambda}}\\
    &a_n(\alpha)=\log\pare{\sum_{\omega\in \Omega}\mu(\omega)e^{u(\alpha,\omega)/\lambda}e^{-b_n(\omega)}}\\
    &\nu_{n+1}(\alpha)=\nu_n(\alpha)e^{a_n(\alpha)}
\end{align*}
By plugging in the first equation into the second, and the second into the third, we arrive at the \textit{Blahut-Arimoto algorithm}:
\begin{align*}
    \nu_{n+1}(\alpha)=\sum_{\omega\in\Omega}  \frac{\mu(\omega)e^{u(\alpha,\omega)/\lambda}\nu_n(\alpha)}{\sum_{\alpha'\in A}e^{u(\alpha',\omega)/\lambda}\nu_n(\alpha')}
\end{align*}
This also shows that each $\nu_n$ is a valid probability in that it is nonnegative and sums to unity. It can be shown that $f(\nu_{n})$ increases in $n$ and, appealing to compactness and continuity, converges to a maximum of $f$. The algorithm does not guarantee convergence of $\nu_n$ \textit{per se}. However, we have guarantees from \cref{th:FOC} that if $\nu_n$ maximizes $f$, then
\begin{align*}
   &\sum_{\omega\in\Omega}  \frac{\mu(\omega)e^{u(\alpha,\omega)/\lambda}}{\sum_{\alpha'\in A}e^{u(\alpha',\omega)/\lambda}\nu_n(\alpha')}=1
\end{align*}
for all $\nu_n(\alpha)>0$, which implies $\nu_{n+1}(\alpha)=\nu_n(\alpha)$. On the other hand, if $\nu_{n+1}=\nu_n=\nu^*$, then
\begin{align*}
    \sum_{\omega\in\Omega}  \frac{\mu(\omega)e^{u(\alpha,\omega)/\lambda}}{\sum_{\alpha'\in A}e^{u(\alpha',\omega)/\lambda}\nu^*(\alpha')}=1
\end{align*}
and if $\nu_n(\alpha)$ is decreasing to zero, then
\begin{align*}
    \sum_{n=0}^\infty\log\pare{\sum_{\omega\in\Omega}  \frac{\mu(\omega)e^{u(\alpha,\omega)/\lambda}}{\sum_{\alpha'\in A}e^{u(\alpha',\omega)/\lambda}\nu_n(\alpha')}}=-\infty
\end{align*}
It should be noted that while each $\nu_n$ is a probability, each candidate potential $a_n$ and $b_n$ are not necessarily Schr\"odinger potentials, and thus do not necessarily define valid information policies via \cref{propertiesofEOT}. Rather, they are smoothmax proxies, and as $n$ grows, $V(\nu_n)$ is increasingly approximated by
\begin{align*}
    V(\nu_n)\approx \sum_{\alpha\in A}a_n(\alpha)\cdot \nu_{n}(\alpha)+ \sum_{\omega\in \Omega} b_n(\omega)\cdot \mu(\omega)
\end{align*}

\section{Conclusion}

We formulate the rational inattention problem as an optimal Schr\"odinger bridgehead problem. We show that the foundational results from rational inattention and their generalizations to continuous spaces follow almost immediately from the theory of entropic optimal transport. Our approach emphasizes the importance of the dual variables---called the \textit{Schr\"odinger potentials}---which are infinite-dimensional analogs of the Lagrange multipliers of the rational inattention problem. Future work could generalize our work to more general classes of cost functions using a more general theory of regularized optimal transport.

\printbibliography

\section*{Appendix}

\begin{proof}[Proof of \cref{gibbsproperty}]
    Let $H$ denote an arbitrary Bayes-plausible probability and $P^\varepsilon=(1-\varepsilon)P+\varepsilon H$. The set of Bayes-plausible probabilities is convex, so the path $P^\varepsilon$ lies in the set. For a dominating measure $\eta\in \Delta A$, re-write
    \begin{align*}
        D_{KL}(P\|P_\alpha\otimes \mu)&=\int \log\pare{\frac{dP}{d(\eta\otimes \mu)}}\ dP - \int \log\pare{\frac{dP_\alpha}{d\eta }}\ dP_\alpha
    \end{align*}
    We have
    \begin{align*}
        \frac{d}{d\varepsilon}\bigg|_{\varepsilon=0} \int \log\pare{\frac{dP^\varepsilon}{d(\eta\otimes \mu)}}\ dP^\varepsilon&=\frac{d}{d\varepsilon}\bigg|_{\varepsilon=0} \int \log\pare{\frac{dP}{d(\eta\otimes \mu)}}\ dP^\varepsilon+\frac{d}{d\varepsilon}\bigg|_{\varepsilon=0} \int \log\pare{\frac{dP^\varepsilon}{d(\eta\otimes \mu)}}\ dP
    \end{align*}
    Evaluating the second term gives us
    \begin{align*}
        \frac{d}{d\varepsilon}\bigg|_{\varepsilon=0} \int \log\pare{\frac{dP^\varepsilon}{d(\eta\otimes \mu)}}\ dP=\int \frac{\frac{d(H-P)}{d(\eta\otimes\mu)}}{\frac{dP}{d(\eta\otimes\mu)}}\ dP=0
    \end{align*}
    And therefore, 
    \begin{align*}
        \frac{d}{d\varepsilon}\bigg|_{\varepsilon=0} \int \log\pare{\frac{dP^\varepsilon}{d(\eta\otimes \mu)}}\ dP^\varepsilon&=\frac{d}{d\varepsilon}\bigg|_{\varepsilon=0} \int \log\pare{\frac{dP}{d(\eta\otimes \mu)}}\ dP^\varepsilon=\int \log\pare{\frac{dP}{d(\eta\otimes \mu)}}\ d(H-P)
    \end{align*}
    Likewise,
    \begin{align*}
        \frac{d}{d\varepsilon}\bigg|_{\varepsilon=0}\int \log\pare{\frac{dP_\alpha^\varepsilon}{d\eta }}\ dP_\alpha^\varepsilon=\int \log\pare{\frac{dP_\alpha}{d\eta }}\ d(H_\alpha-P_\alpha)
    \end{align*}
    And so,
    \begin{align*}
        \frac{d}{d\varepsilon}\bigg|_{\varepsilon=0} D_{KL}(P^\varepsilon\|P^\varepsilon_\alpha\otimes\mu)=\int \log\pare{\frac{dP}{d(P_\alpha\otimes\mu)}}\ d(H-P)
    \end{align*}
    Suppose the integrand is not a $P(\alpha|\omega)$-plateau for $\mu$-a.e. $\omega$. For every $\omega$ where the integrand is not a plateau, define
    \begin{align*}
        A(\omega)=\bigg\{\alpha': \frac{u(\alpha',\omega)}{\lambda}-\log\pare{\frac{dP}{d(P_\alpha\otimes \mu)}(\alpha',\omega)}>\int \frac{u(\alpha,\omega)}{\lambda}-\log\pare{\frac{dP}{d(P_\alpha\otimes \mu)}(\alpha,\omega)}\ dP(\alpha|\omega) \bigg\}
    \end{align*}
    i.e., the set of all points whose integrand is ``above average.'' Define $H(\cdot |\omega)$ such that it is supported only on $A(\omega)$ and define $H(\alpha,\omega)=H(\alpha|\omega)\mu(\omega)$ to arrive at the necessary contradiction.
\end{proof}

\begin{proof}[Proof of \cref{th:FOC}]
    Let $\psi$ denote some other probability over $A$ and consider the Gateaux derivative of $f$ at $\nu$ in the direction of $\psi$. By \cref{lem:propertiesoff},
    \begin{align*}
        \dd_\psi f(\nu)=\int \frac{\int e^{u(\alpha,\omega)/\lambda}\ d\psi(\alpha)}{\int e^{u(\alpha',\omega)/\lambda}\ d\nu(\alpha')}-1\ d\mu(\omega)
    \end{align*}
    which implies that if $\nu$ maximizes $f$, then for any arbitrary $\psi$,
    \begin{align}\label{psileq}
        \int \frac{\int e^{u(\alpha,\omega)/\lambda}\ d\psi(\alpha)}{\int e^{u(\alpha',\omega)/\lambda}\ d\nu(\alpha')}\ d\mu(\omega)\leq 1
    \end{align}
    If there exists $\alpha$ such that the converse of \cref{FOCleq} holds, then then $\psi=\delta_\alpha$ would yield $\dd_\psi f(\nu)>0$. Thus, if $\nu$ maximizes $f$, then \cref{FOCleq}. Integrating the left-hand side of \cref{psileq} with respect to $\nu$ reveals that \cref{FOC} holds $P$-a.s. \\

    Conversely, suppose the first-order conditions hold. Then, for any other $\psi$,
    \begin{align*}
        \pare{\int \frac{\int e^{u(\alpha,\omega)/\lambda}\ d\nu(\alpha)}{\int e^{u(\alpha',\omega)/\lambda}\ d\psi(\alpha')}\ d\mu(\omega)}^{-1}< \int \frac{\int e^{u(\alpha,\omega)/\lambda}\ d\psi(\alpha)}{\int e^{u(\alpha',\omega)/\lambda}\ d\nu(\alpha')}\ d\mu(\omega)\leq 1
    \end{align*}
    unless $\dis \frac{\int e^{u(\alpha,\omega)/\lambda}\ d\psi(\alpha)}{\int e^{u(\alpha',\omega)/\lambda}\ d\nu(\alpha')}$ is constant $\mu$-a.s. Else,
    \begin{align*}
        \dd_\nu f(\psi)=\int \frac{\int e^{u(\alpha,\omega)/\lambda}\ d\nu(\alpha)}{\int e^{u(\alpha',\omega)/\lambda}\ d\psi(\alpha')}-1\ d\mu(\omega)>0
    \end{align*}
\end{proof}

\noindent \textit{Postscript.} This comment originated from the study a different model which could be interpreted as a Shannon RI model with a regularization penalty (of size $\alpha \geq  0$) on the divergence between the action distribution and some ``prior'' over actions. Such a model nests the vanilla Shannon RI model studied in this paper on the boundary $\alpha = 0$. However, it turns out that the model is significantly more well-behaved when $\alpha>0$. Hence, we treat the $\alpha=0$ separately in this comment.

\end{document}